\documentclass{amsart}
\usepackage{amsfonts,amssymb,amsmath,amsthm}
\usepackage{url}
\usepackage{enumerate}
\usepackage{graphicx}
\usepackage{epstopdf}
\usepackage{float}
\usepackage{hyperref}
\usepackage{wrapfig}
\newtheorem{thrm}{Theorem}[section]

\theoremstyle{definition}

\numberwithin{equation}{section}
\author{Dejan Kovacevic}
\address{
21 Neilson Drive
Toronto, M9C 1V3 Canada}
\email{kodza@yahoo.com}
\begin{document}
\title{Analysis of fluid velocity vector field divergence $\nabla\cdot\vec{u}$ in function of variable fluid density $\rho(\vec{x},t)\neq const$ and conditions for vanishing viscosity of compressible Navier-Stokes equations}
\maketitle
\begin{abstract}
In this paper, we perform analysis of the fluid velocity vector field divergence $\nabla \cdot \vec{u}$ derived from the continuity equation, and we explore its application in the Navier-Stokes equations for compressible fluids $\rho (\vec{x},t)\ne const$, occupying all of $\vec{x}\in R^{3} $ space at any $t\geq 0$. The resulting velocity vector field divergence $\nabla \cdot \vec{u}=-\frac{1}{\rho } (\frac{\partial \rho }{\partial t} +\vec{u}\cdot \nabla \rho )$ is a direct consequence of the fluid density rate of change over time $\frac{\partial \rho }{\partial t}$ and over space $\nabla \rho $, in addition to the fluid velocity vector field $\vec{u}(\vec{x},t)$ and the fluid density $\rho (\vec{x},t)$ itself. We derive the conditions for the divergence-free fluid velocity vector field $\nabla \cdot \vec{u}=0$ in scenarios when the fluid density is not constant $\rho(\vec{x},t)\neq const$ over space nor time, and we analyze scenarios of the non-zero divergence $\nabla \cdot \vec{u}\ne 0$. We apply the statement for divergence in the Navier-Stokes equation for compressible fluids, and we deduct the condition for vanishing (zero) viscosity term of the compressible Navier-Stokes equation: $\nabla (\frac{1}{\rho } (\frac{\partial \rho }{\partial t} +\vec{u}\cdot \nabla \rho ))=3\Delta \vec{u}$. In addition to that, we derive even more elementary condition for vanishing viscosity, stating that vanishing viscosity is triggered once scalar function $d(\vec{x},t)=-\frac{1}{\rho } (\frac{\partial \rho }{\partial t} +\vec{u}\cdot \nabla \rho )$ is \textbf{harmonic function}. Once that condition is satisfied, the viscosity related term of the Navier-Stokes equations for compressible fluids equals to zero, which is known as related to turbulent fluid flows.
\end{abstract}
\maketitle
\section{Introduction}
\noindent In this paper, we perform analysis of the fluid velocity vector field divergence $\nabla \cdot \vec{u}$ derived from the continuity equation, and we explore its application in the Navier-Stokes equations for compressible fluids $\rho (\vec{x},t)\ne const$, occupying all of $\vec{x}\in R^{3} $ space at any $t\geq 0$. The divergence-free velocity vector field  $\nabla \cdot \vec{u}=0$ is used in the context of incompressible flows as well as incompressible fluids. In addition to that, the constant fluid density is often used without explicit differentiation between cases of constant fluid density over elapsed time, and constant fluid density over space. Also, the generalized constant fluid density $\rho (\vec{x},t)=const$  over both space and time is often equivalently and interchangeably used with the fluid velocity divergence-free statement $\nabla \cdot \vec{u}=0$. Inconsistent and interchangeable use of terminology with fundamentally different meaning might be a source of misinterpretation and confusion. In this paper, we explore the meaning of the fluid velocity vector field divergence, expressed in function of fluid velocity and density $\nabla \cdot \vec{u}=f(\vec{u},\rho)$.

\noindent We begin analysis with the continuity equation, one of the most fundamental equations in fluid dynamics due to its direct relationship with the law of conservation of mass. From that, we derive the fluid velocity vector field divergence $\nabla \cdot \vec{u}=-\frac{1}{\rho } (\frac{\partial \rho }{\partial t} +\vec{u}\cdot \nabla \rho )$. Based on that statement, we continue analysis of scenarios in which the fluid velocity vector field is divergence-free  $\nabla \cdot \vec{u}=0$ in function of the following two key terms of the velocity divergence statement: $\frac{\partial \rho }{\partial t} $  and $\vec{u}\cdot \nabla \rho $. We also analyze scenarios when divergence  $\nabla \cdot \vec{u}$ is either positive or negative. We share examples demonstrating scenarios when fluid velocity vector field is divergence-free $\nabla \cdot \vec{u}=0$ under conditions when fluid density is not constant and can change over space and/or time. We apply velocity vector field divergence statement  $\nabla \cdot \vec{u}=-\frac{1}{\rho } (\frac{\partial \rho }{\partial t} +\vec{u}\cdot \nabla \rho )$ to the Navier-Stokes equations for compressible fluids resulting with:
\[\frac{\partial \vec{u}}{\partial t} +(\vec{u}\cdot \nabla )\vec{u}=-\frac{1}{\rho } \nabla \bar{p}+\nu (\Delta \vec{u}-\frac{1}{3} \nabla (\frac{1}{\rho } (\frac{\partial \rho }{\partial t} +\vec{u}\cdot \nabla \rho )))+\vec{f}\]
Based on that statement, we derive condition under which the viscosity related term of the Navier-Stokes equations for compressible fluids is equal to zero:
\[\nabla (\frac{1}{\rho } (\frac{\partial \rho }{\partial t} +\vec{u}\cdot \nabla \rho ))=3\Delta \vec{u}\]
Finally, we further deduct even more elementary condition for vanishing (zero) viscosity related term of compressible Navier-Stokes equations.
Derived condition states that vanishing viscosity term of compressible Navier-Stokes equations is triggered once scalar function, representing velocity vector field divergence
\[d(\vec{x},t)=\nabla \cdot \vec{u}=-\frac{1}{\rho } (\frac{\partial \rho }{\partial t} +\vec{u}\cdot \nabla \rho )\]
is \textbf{harmonic function}. Once that condition is satisfied, the viscosity related term of the Navier-Stokes equations for compressible fluids equals to zero, which is known as a condition of turbulent fluid flow.
\begin{thrm}\label{T1}
Let the compressible fluid of density $\rho (\vec{x},t)\in R$ occupy all of $R^{3} $  space, where $\rho (\vec{x},t)$ is smooth and continuously differentiable scalar function, and let $\vec{u}(\vec{x},t)\in R^{3}$ represent the fluid velocity vector field, which is smooth and continuously differentiable for all positions in space $\vec{x}\in R^{3} $  at any time $t\ge 0$.

\noindent Then, the divergence of fluid velocity vector field  $\nabla \cdot \vec{u}$ can be expressed in function of fluid density $\rho (\vec{x},t)$  and the velocity vector field $\vec{u}(\vec{x},t)$  as
\[\nabla \cdot \vec{u}=-\frac{1}{\rho } (\frac{\partial \rho }{\partial t} +\vec{u}\cdot \nabla \rho )=-\frac{1}{\rho } \frac{d\rho }{dt}\]
for all positions in space $\vec{x}\in R^{3} $ at any time $t\ge 0$.
\end{thrm}
\begin{proof}
The continuity equation
\begin{align}
    \frac{\partial \rho }{\partial t} +\nabla \cdot (\rho \vec{u})=0
\label{E1.1}\end{align}
for any $\vec{x}\in R^{3} $ and $t\ge 0$, expresses the law of conservation of mass. From statement (\ref{E1.1}) we further derive:
\begin{align}
    \frac{\partial \rho }{\partial t} +\vec{u}\cdot \nabla \rho +\rho \nabla \cdot \vec{u}=0
\label{E1.2}\end{align}
once terms of the statement (\ref{E1.2}) are re-arranged:
\begin{align}
    \rho \nabla \cdot \vec{u}=-(\frac{\partial \rho }{\partial t} +\vec{u}\cdot \nabla \rho )
\label{E1.3}\end{align}
we conclude
\begin{align}
    \nabla \cdot \vec{u}=-\frac{1}{\rho } (\frac{\partial \rho }{\partial t} +\vec{u}\cdot \nabla \rho )
\label{E1.4}\end{align}
as
\begin{align}
    \vec{u}\cdot \nabla \rho =(u_{x} ,u_{y} ,u_{z} )\cdot (\frac{\partial \rho }{\partial x} ,\frac{\partial \rho }{\partial y} ,\frac{\partial \rho }{\partial z} )
\label{E1.5}\end{align}
\begin{align}
    \vec{u}\cdot \nabla \rho =u_{x} \frac{\partial \rho }{\partial x} +u_{y} \frac{\partial \rho }{\partial y} +u_{z} \frac{\partial \rho }{\partial z}
\label{E1.6}\end{align}
applying (\ref{E1.6}) in (\ref{E1.4}):
\begin{align}
    \nabla \cdot \vec{u}=-\frac{1}{\rho } (\frac{\partial \rho }{\partial t} +u_{x} \frac{\partial \rho }{\partial x} +u_{y} \frac{\partial \rho }{\partial y} +u_{z} \frac{\partial \rho }{\partial z} )
\label{E1.7}\end{align}
the total derivative of density over time is represented as:
\begin{align}
    \frac{d\rho }{dt} =\frac{\partial \rho }{\partial t} +u_{x} \frac{\partial \rho }{\partial x} +u_{y} \frac{\partial \rho }{\partial y} +u_{z} \frac{\partial \rho }{\partial z}
\label{E1.8}\end{align}
Applying (\ref{E1.8}) in (\ref{E1.7}):
\begin{align}
    \nabla \cdot \vec{u}=-\frac{1}{\rho } \frac{d\rho }{dt}
\label{E1.9}\end{align}
combining (\ref{E1.4}) and (\ref{E1.7}) results with
\begin{align}
    \nabla \cdot \vec{u}=-\frac{1}{\rho } (\frac{\partial \rho }{\partial t} +\vec{u}\cdot \nabla \rho )=-\frac{1}{\rho } \frac{d\rho }{dt}
\label{E1.10}\end{align}
for any $\vec{x}\in R^{3} $ and $t\ge 0$, which proves this theorem.
\end{proof}
\noindent Each of two terms $\frac{\partial \rho }{\partial t} $ and  $\vec{u}\cdot \nabla \rho $ of the statement for fluid velocity vector field divergence $\nabla \cdot \vec{u}=-\frac{1}{\rho } (\frac{\partial \rho }{\partial t} +\vec{u}\cdot \nabla \rho )$  might be zero, positive, or negative for any selected position in space $x\in R^{3} $  at any time $t\ge 0$. Therefore, the sum of the terms $\frac{\partial \rho }{\partial t} $  and  $\vec{u}\cdot \nabla \rho $ within the brackets, could be negative, positive, or zero. In case that both terms, $\frac{\partial \rho }{\partial t} $ and $\vec{u}\cdot \nabla \rho $, are of equal value and opposite sign, their sum is zero, resulting in zero divergence $\nabla \cdot \vec{u}=0$. This means that fluid density  $\rho $ can change over time as well as over space, while the resulting velocity vector field can be divergence-free $\nabla \cdot \vec{u}=0$  for as long as the sum of the terms, $\frac{\partial \rho }{\partial t} $ and  $\vec{u}\cdot \nabla \rho $, results in zero.

\noindent Notably, the divergence of fluid velocity could be zero  $\nabla \cdot \vec{u}=0$ also in the case when the gradient of density is zero $\nabla \rho =0$  and the density change over time is also zero $\frac{\partial \rho }{\partial t} =0$  for any position $x\in R^{3} $  and $t\ge 0$  .  However, such a scenario represents a special case only in comparison with all possible scenarios.

\noindent Let us explore, in detail, the following scenarios, and their impact on fluid velocity vector field divergence:

\noindent 1)$\frac{\partial \rho }{\partial t} >0$  and  $\vec{u}\cdot \nabla \rho >0$

\noindent then
\[\frac{\partial \rho }{\partial t} +\vec{u}\cdot \nabla \rho >0\]
therefore
\[\nabla \cdot \vec{u}=-\frac{1}{\rho } (\frac{\partial \rho }{\partial t} +\vec{u}\cdot \nabla \rho )<0\]
In this scenario, divergence of fluid velocity vector field is negative $\nabla \cdot \vec{u}<0$.

\noindent 2) $\frac{\partial \rho }{\partial t} >0$ and  $\vec{u}\cdot \nabla \rho <0$and $\left|\frac{\partial \rho }{\partial t} \right|>\left|\vec{u}\cdot \nabla \rho \right|$

\noindent then
\[\frac{\partial \rho }{\partial t} +\vec{u}\cdot \nabla \rho >0\]
therefore
\[\nabla \cdot \vec{u}=-\frac{1}{\rho } (\frac{\partial \rho }{\partial t} +\vec{u}\cdot \nabla \rho )<0\]
In this scenario, divergence of fluid velocity vector field is negative $\nabla \cdot \vec{u}<0$.

\noindent 3) $\frac{\partial \rho }{\partial t} >0$ and  $\vec{u}\cdot \nabla \rho <0$and $\left|\frac{\partial \rho }{\partial t} \right|<\left|\vec{u}\cdot \nabla \rho \right|$

\noindent then
\[\frac{\partial \rho }{\partial t} +\vec{u}\cdot \nabla \rho <0\]
therefore
\[\nabla \cdot \vec{u}=-\frac{1}{\rho } (\frac{\partial \rho }{\partial t} +\vec{u}\cdot \nabla \rho )>0\]
In this scenario, divergence of fluid velocity vector field is positive $\nabla \cdot \vec{u}>0$.

\noindent 4) $\frac{\partial \rho }{\partial t} <0$ and  $\vec{u}\cdot \nabla \rho >0$and $\left|\frac{\partial \rho }{\partial t} \right|>\left|\vec{u}\cdot \nabla \rho \right|$

\noindent then
\[\frac{\partial \rho }{\partial t} +\vec{u}\cdot \nabla \rho <0\]
therefore
\[\nabla \cdot \vec{u}=-\frac{1}{\rho } (\frac{\partial \rho }{\partial t} +\vec{u}\cdot \nabla \rho )>0\]
In this scenario, divergence of fluid velocity vector field is positive $\nabla \cdot \vec{u}>0$.

\noindent 5) $\frac{\partial \rho }{\partial t} <0$ and  $\vec{u}\cdot \nabla \rho >0$and $\left|\frac{\partial \rho }{\partial t} \right|<\left|\vec{u}\cdot \nabla \rho \right|$

\noindent then
\[\frac{\partial \rho }{\partial t} +\vec{u}\cdot \nabla \rho >0\]
therefore
\[\nabla \cdot \vec{u}=-\frac{1}{\rho } (\frac{\partial \rho }{\partial t} +\vec{u}\cdot \nabla \rho )<0\]
In this scenario, divergence of fluid velocity vector field is negative $\nabla \cdot \vec{u}<0$.

\noindent 6) $\frac{\partial \rho }{\partial t} <0$ and  $\vec{u}\cdot \nabla \rho <0$

\noindent then
\[\frac{\partial \rho }{\partial t} +\vec{u}\cdot \nabla \rho <0\]
therefore
\[\nabla \cdot \vec{u}=-\frac{1}{\rho } (\frac{\partial \rho }{\partial t} +\vec{u}\cdot \nabla \rho )>0\]
In this scenario, divergence of fluid velocity vector field is positive $\nabla \cdot \vec{u}>0$.

\begin{thrm}\label{T2}
Let the compressible fluid of density $\rho (\vec{x},t)\in R$ occupy all of $R^{3} $  space, where $\rho (\vec{x},t)$ is smooth and continuously differentiable scalar function, and let $\vec{u}(\vec{x},t)\in R^{3}$ represent the fluid velocity vector field, which is smooth and continuously differentiable for all positions in space $\vec{x}\in R^{3} $  at any time $t\ge 0$.

\noindent Then, the fluid velocity vector field $\vec{u}(\vec{x},t)$  is divergence-free $\nabla \cdot \vec{u}=0$  only if following condition is satisfied:
\[\frac{\partial \rho }{\partial t} =-\vec{u}\cdot \nabla \rho \]
for all positions $\vec{x}\in R^{3} $ at any time $t\ge 0$.
\end{thrm}
\begin{proof}
\noindent As per Theorem \ref{T1}:
\begin{align}
    \nabla \cdot \vec{u}=-\frac{1}{\rho } (\frac{\partial \rho }{\partial t} +\vec{u}\cdot \nabla \rho )
\label{E2.1}\end{align}
as per statement of this theorem:
\begin{align}
    \nabla \cdot \vec{u}=0
\label{E2.2}\end{align}
applying (\ref{E2.2}) on (\ref{E2.1}):
\begin{align}
    -\frac{1}{\rho } (\frac{\partial \rho }{\partial t} +\vec{u}\cdot \nabla \rho )=0
\label{E2.3}\end{align}
\begin{align}
    \frac{\partial \rho }{\partial t} +\vec{u}\cdot \nabla \rho =0
\label{E2.4}\end{align}
rearranging terms of (\ref{E2.4}) we conclude
\begin{align}
    \frac{\partial \rho }{\partial t} =-\vec{u}\cdot \nabla \rho
\label{E2.5}\end{align}
for all positions $\vec{x}\in R^{3} $ at any time $t\ge 0$, which proves this theorem.
\end{proof}

\noindent \textbf{Example 1.2.1}

\noindent In this example we will select fluid density which is variable both in space and over time. We will demonstrate that velocity divergence is zero although fluid is compressible and density is not constant over space nor time.

\noindent Let's define fluid density as:
\begin{align}
    \rho =2+\cos (x+y+z+\cos (t))
\label{E2.9}\end{align}
for all positions $\vec{x}\in R^{3} $ at any time $t\ge 0$, and fluid velocity as:
\begin{align}
    \vec{u}=\frac{1}{3} (\sin (t),\sin (t),\sin (t))
\label{E2.10}\end{align}
for all positions $\vec{x}\in R^{3} $ at any time $t\ge 0$.
As per Theorem \ref{T1}:
\begin{align}
    \nabla \cdot \vec{u}=-\frac{1}{\rho } (\frac{\partial \rho }{\partial t} +\vec{u}\cdot \nabla \rho )
\label{E2.11}\end{align}
Let's apply $\rho$ as per (\ref{E2.9}) in $\frac{\partial \rho }{\partial t}$, representing the first term in the brackets of statement (\ref{E2.11}):
\begin{align}
    \frac{\partial \rho }{\partial t} =\frac{\partial (2+\cos (x+y+z+\cos (t)))}{\partial t}
\label{E2.12}\end{align}
\begin{align}
    \frac{\partial \rho }{\partial t} =\sin (t)\sin (x+y+z+\cos (t))
\label{E2.13}\end{align}
applying gradient $\nabla$ operator on fluid density $\rho $ as per (\ref{E2.9}):
\begin{align}
    \nabla \rho =\nabla (2+\cos (x+y+z+\cos (t)))
\label{E2.14}\end{align}
\begin{align}
    \nabla \rho =(-\sin (x+y+z+\cos (t)),-\sin (x+y+z+\cos (t)),-\sin (x+y+z+\cos (t)))
\label{E2.15}\end{align}
applying on (\ref{E2.15}) dot product with fluid velocity vector field $\vec{u}$ as per (\ref{E2.10}):
\begin{align}
    \vec{u}\cdot \nabla \rho =\frac{1}{3} (\sin (t),\sin (t),\sin (t))\cdot
\label{E2.16}\end{align}
\[(-\sin (x+y+z+\cos (t)),-\sin (x+y+z+\cos (t)),-\sin (x+y+z+\cos (t)))\]
\begin{align}
    \vec{u}\cdot \nabla \rho =-\sin (t)\sin (x+y+z+\cos (t))
\label{E2.16.1}\end{align}
applying (\ref{E2.13}) and (\ref{E2.16.1})in (\ref{E2.11}):
\begin{align}
    \nabla \cdot \vec{u}=-\frac{1}{\rho } (\sin (t)\sin (x+y+z+\cos (t))-\sin (t)\sin (x+y+z+\cos (t)))
\label{E2.17}\end{align}
\begin{align}
    \nabla \cdot \vec{u}=0
\label{E2.19}\end{align}
which demonstrates that vector field is divergence-free, in scenario when fluid density is variable both over space as well as over time.

\begin{thrm}\label{T3}
Let the compressible fluid of density $\rho (\vec{x},t)\in R$ occupy all of $R^{3} $  space, where $\rho (\vec{x},t)$ is smooth and continuously differentiable scalar function, and let $\vec{u}(\vec{x},t)\in R^{3}$ represent the fluid velocity vector field, which is smooth and continuously differentiable for all positions in space $\vec{x}\in R^{3} $  at any time $t\ge 0$. Also, let $\alpha (\vec{x},t)$ represents angle between two vectors belonging to the vector fields $\vec{u}(\vec{x},t)\in R^{3}$ and $\nabla \rho (\vec{x},t)\in R^{3}$ respectively, sharing the same position in space $\vec{x}\in R^{3}$ at any time $t\ge 0$.

\noindent Then, in case that following conditions are satisfied:
\[\frac{\partial \rho }{\partial t} =0 ;  \nabla \rho \ne 0; \alpha(\vec{x},t) =\pm \frac{\pi }{2} \]

\noindent fluid velocity vector field $\vec{u}$ must be divergence-free $\nabla \cdot \vec{u}=0$ for all positions $\vec{x}\in R^{3}$  at any time $t\ge 0$.
\end{thrm}
\begin{proof}
\noindent As per Theorem \ref{T1}:
\begin{align}
    \nabla \cdot \vec{u}=-\frac{1}{\rho } (\frac{\partial \rho }{\partial t} +\vec{u}\cdot \nabla \rho )
\label{E3.01}\end{align}
for any $\vec{x}\in R^{3}$  at $t\ge 0$.
Let's observe term $\vec{u}\cdot \nabla \rho$ in statement (\ref{E3.01}), which could be represented in following way:
\begin{align}
    \vec{u}\cdot \nabla \rho =\left|\vec{u}\right|\left|\nabla \rho \right|\cos \alpha(\vec{x},t)
\label{E3.1}\end{align}
for any $\vec{x}\in R^{3}$  at $t\ge 0$.
As per condition defined by this theorem:
\begin{align}
    \cos \alpha(\vec{x},t)=\cos \left(\pm \frac{\pi }{2} \right)=0
\label{E3.2}\end{align}
for any $\vec{x}\in R^{3}$  at $t\ge 0$.
Applying (\ref{E3.2}) in (\ref{E3.1}):
\begin{align}
    \vec{u}\cdot \nabla \rho =0
\label{E3.3}\end{align}
for any $\vec{x}\in R^{3}$  at $t\ge 0$.
As per this theorem statement:
\begin{align}
    \frac{\partial \rho }{\partial t} =0
\label{E3.4}\end{align}
for any $\vec{x}\in R^{3}$  at $t\ge 0$.
Applying (\ref{E3.3}) and (\ref{E3.4}) in (\ref{E3.01}):
\begin{align}
    \nabla \cdot \vec{u}=0
\label{E3.6}\end{align}
for all positions $\vec{x}\in R^{3}$ at any time $t\ge 0$, which proves this theorem.
\end{proof}

\noindent \textbf{Example 1.3.1}
\noindent This example will demonstrate one scenario, in which fluid velocity vector field $\vec{u}$ is divergence-free $\nabla\cdot\vec{u}=0$  when fluid velocity vector field divergence is expressed, as per Theorem \ref{T1}, in function of fluid density:
 \begin{align}
    \nabla \cdot \vec{u}=-\frac{1}{\rho } (\frac{\partial \rho }{\partial t} +\vec{u}\cdot \nabla \rho )
\label{E3.6.1}\end{align}
Fluid density gradient, in this example, is to be selected such that it is not constant $\nabla\rho \neq const$. What that means is that fluid density changes with change of position in space, for any selected time $t\geq 0$. On the other hand, fluid density $\rho$  to be selected, is constant over any elapsed time starting from any time $t\geq 0$ and for any position in space $\vec{x}\in R^{3}$. As per Theorem \ref{T3}, for velocity vector field $\vec{u}$ to be divergence-free $\nabla\cdot\vec{u}=0$, vectors belonging to the gradient of density vector field $\nabla\rho$ must be orthogonal to vectors belonging to the fluid velocity vector field $\vec{u}$, for each two vectors sharing the same position in space $\vec{x}\in R^{3}$ and at the same time $t\ge 0$.

\noindent Let's define fluid density as:
\begin{align}
    \rho(\vec{x}) =\frac{1}{1+\frac{1}{x^{2} +y^{2} +z^{2} } }
\label{E3.7}\end{align}
at any position $\vec{x}\in R^{3}$ at any time $t\ge 0$.
As density is not in function of time, partial derivative by time must be zero:
\begin{align}
    \frac{\partial \rho(\vec{x})}{\partial t} =0
\label{E3.8}\end{align}
applying gradient on fluid density $\rho(\vec{x})$:
\begin{align}
    \nabla \rho(\vec{x)} =\frac{2}{\left(x^{2} +y^{2} +z^{2} \right)^{2} \left(1+\frac{1}{x^{2} +y^{2} +z^{2} } \right)^{2} } (x,y,z)
\label{E3.9}\end{align}
Let's define fluid velocity vector field $\vec{u}(\vec{x)}$ as:
\begin{align}
    \vec{u}=(-y,x,0)
\label{E3.10}\end{align}
at any position $\vec{x}\in R^{3}$ at any time $t\ge 0$.
Applying dot product between vector fields $\nabla \rho$ and $\vec{u}$:
\begin{align}
    \vec{u}\cdot \nabla \rho =(-y,x,0)\cdot (x,y,z)\frac{1}{\left(x^{2} +y^{2} +z^{2} \right)^{2} \left(1+\frac{1}{x^{2} +y^{2} +z^{2} } \right)^{2} }
\label{E3.11}\end{align}
\begin{align}
    \vec{u}\cdot \nabla \rho =\frac{-xy+xy}{\left(x^{2} +y^{2} +z^{2} \right)^{2} \left(1+\frac{1}{x^{2} +y^{2} +z^{2} } \right)^{2} }
\label{E3.12}\end{align}
\begin{align}
    \vec{u}\cdot \nabla \rho =0
\label{E3.13}\end{align}
per Theorem \ref{T1}
\begin{align}
    \nabla \cdot \vec{u}=-\frac{1}{\rho } (\frac{\partial \rho }{\partial t} +\vec{u}\cdot \nabla \rho )
\label{E3.14}\end{align}
applying (\ref{E3.8}) and (\ref{E3.13}) in (\ref{E3.14}):
\begin{align}
    \nabla \cdot \vec{u}=-\frac{1}{\rho } (0+0)
\label{E3.15}\end{align}
\begin{align}
    \nabla \cdot \vec{u}=0
\label{E3.16}\end{align}
for all positions $\vec{x}\in R^{3}$  at any time $t\ge 0$.
In addition to that, in case that fluid velocity vector field divergence is derived based on direct application of divergence operator $\nabla \cdot$ over velocity vector field $\vec{u}$:
\begin{align}
    \nabla \cdot \vec{u}=\nabla \cdot (-y,x,0)=0
\label{E3.17}\end{align}
same result is obtained.

\noindent This example demonstrated that fluid velocity vector field is divergence-free  $\nabla \cdot \vec{u}=0$, for selected fluid density $\rho$ such that it is constant over elapsed time, for all positions in space $\vec{x}\in R^{3}$, and which is not constant across all positions in space $\vec{x}\in R^{3}$ for any $t\geq 0$.
\begin{thrm}\label{T4}
Let the compressible fluid of density $\rho (\vec{x},t)\in R$ occupy all of $R^{3} $  space, where $\rho (\vec{x},t)$ is smooth and continuously differentiable scalar function, and let $\vec{u}(\vec{x},t)\in R^{3}$ represent the fluid velocity vector field, which is smooth and continuously differentiable for all positions in space $\vec{x}\in R^{3} $  at any time $t\ge 0$.

\noindent Then Navier-Stokes equation for compressible fluid can be expressed as:
\[\frac{\partial \vec{u}}{\partial t} +(\vec{u}\cdot \nabla )\vec{u}=-\frac{1}{\rho } \nabla \bar{p}+\nu (\Delta \vec{u}-\frac{1}{3} \nabla (\frac{1}{\rho } (\frac{\partial \rho }{\partial t} +\vec{u}\cdot \nabla \rho )))+\vec{f}\]
for all positions $\vec{x}\in R^{3}$  at any time $t\ge 0$.
\end{thrm}
\begin{proof}
\noindent Starting from statement for compressible Navier-Stokes equation:
\begin{align}
    \frac{\partial \vec{u}}{\partial t} +(\vec{u}\cdot \nabla )\vec{u}=-\frac{1}{\rho } \nabla \bar{p}+\nu (\Delta \vec{u}+\frac{1}{3} \nabla (\nabla \cdot \vec{u}))+\vec{f}
\label{E4.1}\end{align}
for all positions $\vec{x}\in R^{3}$ at any time $t\ge 0$,
as per Theorem \ref{T1}:
\begin{align}
    \nabla \cdot \vec{u}=-\frac{1}{\rho } (\frac{\partial \rho }{\partial t} +\vec{u}\cdot \nabla \rho )
\label{E4.2}\end{align}
for all positions $\vec{x}\in R^{3}$ at any time $t\ge 0$,
applying (\ref{E4.2}) in (\ref{E4.1}):
\begin{align}
    \frac{\partial \vec{u}}{\partial t} +(\vec{u}\cdot \nabla )\vec{u}=-\frac{1}{\rho } \nabla \bar{p}+\nu (\Delta \vec{u}-\frac{1}{3} \nabla (\frac{1}{\rho } (\frac{\partial \rho }{\partial t} +\vec{u}\cdot \nabla \rho )))+\vec{f}
\label{E4.3}\end{align}
for all positions $\vec{x}\in R^{3}$  at any time $t\ge 0$, which proves this Theorem.
\end{proof}

\begin{thrm}\label{T5}
Let the compressible fluid of density $\rho (\vec{x},t)\in R$ occupy all of $R^{3} $  space, where $\rho (\vec{x},t)$ is smooth and continuously differentiable scalar function, and let $\vec{u}(\vec{x},t)\in R^{3}$ represent the fluid velocity vector field, which is smooth and continuously differentiable for all positions in space $\vec{x}\in R^{3} $  at any time $t\ge 0$.

\noindent Then condition for vanishing (zero) viscosity related term of compressible Navier-Stokes equation, can be expressed as:
\[\nabla (\frac{1}{\rho } (\frac{\partial \rho }{\partial t} +\vec{u}\cdot \nabla \rho ))=3\Delta \vec{u}\]
for all positions $\vec{x}\in R^{3} $ at any time$t\ge 0$.
\end{thrm}
\begin{proof}
\noindent As per Theorem \ref{T4} Navier-Stokes equation for compressible fluid can be expressed as:
\begin{align}
    \frac{\partial \vec{u}}{\partial t} +(\vec{u}\cdot \nabla )\vec{u}=-\frac{1}{\rho } \nabla \bar{p}+\nu (\Delta \vec{u}-\frac{1}{3} \nabla (\frac{1}{\rho } (\frac{\partial \rho }{\partial t} +\vec{u}\cdot \nabla \rho )))+\vec{f}
\label{E5.1}\end{align}
for all positions $\vec{x}\in R^{3} $ at any time $t\ge 0$.

\noindent Let's explore conditions when viscosity related term in statement (\ref{E5.1}) is zero:
\begin{align}
    \nu (\Delta \vec{u}-\frac{1}{3} \nabla (\frac{1}{\rho } (\frac{\partial \rho }{\partial t} +\vec{u}\cdot \nabla \rho )))=0
\label{E5.2}\end{align}
\begin{align}
    \Delta \vec{u}-\frac{1}{3} \nabla (\frac{1}{\rho } (\frac{\partial \rho }{\partial t} +\vec{u}\cdot \nabla \rho )))=0
\label{E5.3}\end{align}
\begin{align}
    \Delta \vec{u}=\frac{1}{3} \nabla (\frac{1}{\rho } (\frac{\partial \rho }{\partial t} +\vec{u}\cdot \nabla \rho )))
\label{E5.4}\end{align}
\begin{align}
    \nabla (\frac{1}{\rho } (\frac{\partial \rho }{\partial t} +\vec{u}\cdot \nabla \rho )))=3\Delta \vec{u}
\label{E5.5}\end{align}
for all positions $\vec{x}\in R^{3} $ at any time $t\ge 0$, which proves this theorem.
\end{proof}
Based on Theorem \ref{T5}, we conclude that in case when fluid density and fluid velocity are such that statement
\[\nabla (\frac{1}{\rho } (\frac{\partial \rho }{\partial t} +\vec{u}\cdot \nabla \rho )))=3\Delta \vec{u}\]
is satisfied, viscosity related term of Navier-Stokes equation for compressible fluids
\[\nu (\Delta \vec{u}-\frac{1}{3} \nabla (\frac{1}{\rho } (\frac{\partial \rho }{\partial t} +\vec{u}\cdot \nabla \rho )))\]
becomes equal to zero.
\[\nu (\Delta \vec{u}-\frac{1}{3} \nabla (\frac{1}{\rho } (\frac{\partial \rho }{\partial t} +\vec{u}\cdot \nabla \rho )))=0\]
\noindent Based on that, it could be concluded that once such condition is satisfied, as per Navier-Stokes equation for compressible fluids and Theorem \ref{T5}, viscosity should not have any effect on dynamics of fluid.

\begin{thrm}\label{T6}
Let the compressible fluid of density $\rho (\vec{x},t)\in R$ occupy all of $R^{3} $  space, where $\rho (\vec{x},t)$ is smooth and continuously differentiable scalar function, and let $\vec{u}(\vec{x},t)\in R^{3}$ represent the fluid velocity vector field, which is smooth and continuously differentiable for all positions in space $\vec{x}\in R^{3} $  at any time $t\ge 0$.

\noindent Then condition for vanishing (zero) viscosity related term of compressible Navier-Stokes equation for all $\vec{x}\in R^{3} $ and $t\ge 0$ is that scalar function
\[d(\vec{x},t)=-\frac{1}{\rho } (\frac{\partial \rho }{\partial t} +\vec{u}\cdot \nabla \rho )\]
\textbf{mush be harmonic function} satisfying Laplace's equation
\[\Delta d(\vec{x},t)=0\]
 for all positions $\vec{x}\in R^{3} $ at any time $t\ge 0$.
\end{thrm}
\begin{proof}
\noindent Starting from statement for compressible Navier-Stokes equation:
\begin{align}
    \frac{\partial \vec{u}}{\partial t} +(\vec{u}\cdot \nabla )\vec{u}=-\frac{1}{\rho } \nabla \bar{p}+\nu (\Delta \vec{u}+\frac{1}{3} \nabla (\nabla \cdot \vec{u}))+\vec{f}
\label{E6.1}\end{align}
viscosity related term of Navier-Stokes equation (\ref{E6.1}) must be zero for vanishing viscosity:
\begin{align}
    \nu (\Delta \vec{u}+\frac{1}{3} \nabla (\nabla \cdot \vec{u}))=0
\label{E6.2}\end{align}
\begin{align}
    \Delta \vec{u}+\frac{1}{3} \nabla (\nabla \cdot \vec{u})=0
\label{E6.3}\end{align}
applying  divergence operator $\nabla \cdot$ on (\ref{E6.3}):
\begin{align}
    \nabla \cdot \Delta \vec{u}+\frac{1}{3} \nabla \cdot \nabla (\nabla \cdot \vec{u})=0
\label{E6.4}\end{align}
as
\begin{align}
    \nabla \cdot \Delta \vec{u}=\Delta(\nabla \cdot \vec{u})
\label{E6.5}\end{align}
and as
\begin{align}
    \nabla \cdot \nabla (\nabla \cdot \vec{u})=\Delta (\nabla \cdot \vec{u})
\label{E6.6}\end{align}
applying (\ref{E6.5}) and (\ref{E6.6}) in (\ref{E6.4}):
\begin{align}
    \Delta (\nabla \cdot \vec{u})+\frac{1}{3} \Delta (\nabla \cdot \vec{u})=0
\label{E6.7}\end{align}
\begin{align}
    \frac{4}{3} \Delta(\nabla \cdot \vec{u})=0
\label{E6.8}\end{align}
\begin{align}
    \Delta(\nabla \cdot \vec{u})=0
\label{E6.9}\end{align}
let's define scalar function
\begin{align}
    d(\vec{x},t)=\nabla \cdot \vec{u}
\label{E6.10}\end{align}
for all positions $\vec{x}\in R^{3} $ at any time $t\ge 0$.
\noindent Applying (\ref{E6.10}) in (\ref{E6.9}):
\begin{align}
    \Delta d(\vec{x},t)=0
\label{E6.11}\end{align}
for all positions $\vec{x}\in R^{3} $ at any time $t\ge 0$.
\noindent Statement (\ref{E6.11}) represents Laplace's equation, which has solutions when scalar function
\begin{align}
    d(\vec{x},t)=\nabla \cdot \vec{u}
\label{E6.12}\end{align}
\textbf{is harmonic function} for all positions $\vec{x}\in R^{3}$ at any time $t\ge 0$.
\noindent As per Theorem 1:
\begin{align}
    \nabla \cdot \vec{u}=-\frac{1}{\rho } (\frac{\partial \rho }{\partial t} +\vec{u}\cdot \nabla \rho )
\label{E6.13}\end{align}
Applying (\ref{E6.13}) in (\ref{E6.12}) we conclude that
\begin{align}
    d(\vec{x},t)=-\frac{1}{\rho } (\frac{\partial \rho }{\partial t} +\vec{u}\cdot \nabla \rho )
\label{E6.14}\end{align}
\textbf{must be harmonic function} for all positions $\vec{x}\in R^{3} $ at any time $t\ge 0$.
\noindent Based on (\ref{E6.14}) and (\ref{E6.11}) we conclude that this theorem is proven.
\end{proof}

\section{Discussion}
\noindent Based on performed analysis, we conclude that fluid velocity vector field divergence $ \nabla \cdot \vec{u}$ , expressed in function of fluid density and fluid velocity:
\[\nabla \cdot \vec{u}=-\frac{1}{\rho } (\frac{\partial \rho }{\partial t} +\vec{u}\cdot \nabla \rho )\]
\noindent describes essential relationship between fluid velocity vector field divergence and fluid density, which variability over space as well as over time, directly influences and in essence, defines velocity divergence.
\noindent Statement of Theorem \ref{T1}, outlines the essence of the fluid velocity vector field divergence $\nabla \cdot \vec{u}$ as a result of the direct derivation from one of the most fundamental statements in fluid dynamics, continuity equation, representing the law of conservation of mass. The continuity equation ensures that mass of fluid is conserved in all circumstances. From such fundamental statement, Theorem \ref{T1} demonstrates derivation of the divergence of fluid velocity vector field $\nabla \cdot \vec{u}$  which is in function of the fluid density, fluid density rate of change over time $\frac{\partial \rho }{\partial t}$ and space $\nabla \rho$, in addition to fluid velocity vector field itself $\vec{u}$.
\noindent The statement
\[\nabla \cdot \vec{u}=-\frac{1}{\rho } (\frac{\partial \rho }{\partial t} +\vec{u}\cdot \nabla \rho )=-\frac{1}{\rho } \frac{d\rho }{dt}\]
in addition to representing relationship between velocity vector field divergence and fluid density, provides the insight into the actual meaning behind velocity vector field divergence. As per that statement, velocity vector field divergence represents ratio between total derivative of fluid density over time and density itself.
Total derivative of fluid density, actually represents average density of fluid change over some elapsed time $dt$. Once $\frac{d\rho}{dt}$ is divided with density $\rho$ with negative sign, resulting statement represents velocity vector field divergence. Such ratio is negative when total derivative of fluid density over time is positive $\frac{d\rho}{dt}>0$. What that means is that for observed infinitesimal volume of space occupied by fluid, average fluid density is increased over elapsed increment of time, while at the same time, as consequence, resulting velocity vector field divergence is negative. In order for average fluid density to be increased within the observed infinitesimal volume of space over an increment of time $dt$, some quantity of fluid has to be added to the observed infinitesimal volume of space, which can happen over the surface of observed volume of space, through which, additional fluid flows into the infinitesimal volume, in order to contribute to increasing fluid density within the volume. Therefore, conclusion that divergence is negative in case when fluid density is increased makes intuitive sense.
On the other hand, in case that fluid total density within infinitesimal volume of space is decreased over some increment in time $dt$ means that some quantity of fluid, which used to be within that volume, left the observed infinitesimal volume of space, as its density is smaller compared to the time prior to the elapsed increment of time $dt$. What that means is that some quantity of fluid traveled through surfaces of volume of space observed, from within observed infinitesimal volume towards its exterior, making resulting divergence positive, which also makes intuitive sense.
\newline
We conclude that fluid velocity vector field divergence $\nabla \cdot \vec{u}$ can be intuitively interpreted as normalized rate of change of averaged fluid density within an infinitesimal volume of space, over an increment of elapsed time, with negative sign. When positive, indicates loss of average fluid density within the observed infinitesimal volume of space over an increment of time $dt$. On the other hand, when negative, it indicates increase of fluid density within the observed infinitesimal volume of space and over an increment of time $dt$.
In case that rate of density change over an increment of time $dt$ is zero, then divergence is also zero, meaning that average density of fluid within volume of space is constant over elapsed time.

In addition to that, we demonstrated examples when fluid velocity vector field $\vec{u}$ can be divergence-free $\nabla \cdot \vec{u}=0$ in scenarios when fluid density $\rho$ is not constant over space and/or over time. By applying the velocity vector field divergence in such form to the Navier-Stokes equation for compressible fluids, we obtained compressible Navier-Stokes equation in following form:
\[\frac{\partial \vec{u}}{\partial t} +(\vec{u}\cdot \nabla )\vec{u}=-\frac{1}{\rho } \nabla \bar{p}+\nu (\Delta \vec{u}-\frac{1}{3} \nabla (\frac{1}{\rho } (\frac{\partial \rho }{\partial t} +\vec{u}\cdot \nabla \rho )))+\vec{f}\]
which can be utilized both for scenarios when fluid velocity vector field is not divergence-free, as well as when it is.
Based on such modified compressible Navier-Stokes equation, condition for vanishing viscosity term, in function of fluid density and velocity can be expressed as:
\[\nabla (\frac{1}{\rho } (\frac{\partial \rho }{\partial t} +\vec{u}\cdot \nabla \rho ))=3\Delta \vec{u}\]
In addition to that, we deduct even more elementary condition for vanishing viscosity term of compressible Navier-Stokes equation, stating that scalar function
\[d(\vec{x},t)=-\frac{1}{\rho } (\frac{\partial \rho }{\partial t} +\vec{u}\cdot \nabla \rho )\]
\textbf{mush be harmonic, satisfying Laplace's equation}
\[\Delta d(\vec{x},t)=0\]
for all positions $\vec{x}\in R^{3} $ at any time $t\ge 0$.
\newline
Such conclusion, might represent an opportunity for better understanding conditions leading to triggering turbulent behaviour in fluids. Deducted condition for vanishing viscosity, in its essence, is in function of fluid density rate of change across space as well as time. When rates of change of fluid density over space, in form of gradient of density $\nabla \rho$, and fluid density rate of change over time, in form of $\frac{\partial \rho }{\partial t}$ are combined, resulting scalar function, once becoming harmonic over space and time, represents a condition for vanishing viscosity related term in compressible Navier-Stokes equation. Such realization might provide new opportunities for better understanding role of fluid density in triggering turbulence, as well as viscosity and core underlying physical mechanisms behind it.

\end{document}